%% file: main.tex
\documentclass[conference]{IEEEtran}
\usepackage{notoccite}
\usepackage[numbers]{natbib}

\title{Joint Methods For Traffic Engineering and Oscillation Repair}

\newcommand{\eat}[1]{}

\eat{
\author[1]{Alexander Gurney \thanks{alexander.gurney@gmail.com}}
\author[1]{Behnaz Arzani\thanks{barzani@seas.upenn.edu}}
\author{Bo Li}
\author{Xianglong Han}
\author[2]{Roch Guerin\thanks{guerin@wustl.edu}}
\author[1]{Boon Thau Loo\thanks{boonloo@cis.upenn.edu}}
\affil[1]{University Of Pennsylvania}
\affil[2]{University of Washington at St. Louis}
\affil[3]{ Comcast }

}

\date{}

\usepackage[svgnames]{xcolor}

\usepackage[section]{placeins}
\usepackage{url}
\usepackage{amsmath}
\usepackage{tikz}
\usepackage{subcaption}
\usepackage{graphicx}
\graphicspath{{images/}}

\usepackage{amsthm}
\theoremstyle{plain}

\newtheorem{definition}{Definition}
\newtheorem{lemma}{Lemma}

\begin{document}

\maketitle

\begin{abstract}
\input{abstract.tex}

\end{abstract}

\input{intro.tex}

\input{probstatement.tex}

\input{overview.tex}
\input{PSPP.tex}

\input{repair.tex}

\input{evaluation.tex}

\begin{scriptsize}
\bibliographystyle{abbrv}
\bibliography{biblio}
\end{scriptsize}
\end{document}

%% file: abstract.tex
The performance of networks that use the Internet Protocol stack is sensitive
to precise configuration of many low-level parameters on each network device. 
These settings govern the action of dynamic routing protocols, which direct the flow of
traffic; in order to ensure that these dynamic protocols all converge to produce some 'optimal' flow, each 
parameter must be set correctly. Multiple conflicting optimization objectives, nondeterminism, and the need to reason
about different failure scenarios make the task particularly complicated. 

We present a fast and flexible approach for the analysis of a number of such management tasks presented in the context of BGP routing. The idea is to combine logical satisfiability criteria with traditional numerical optimization, to reach a desired traffic 
flow outcome subject to given constraints on the routing process. The method can then be used to probe parameter sensitivity, trade-offs in the selection of optimization goals, resilience to failure, and so forth. 
The theory is underpinned by a rigorous abstraction of the convergence of distributed asynchronous message-passing protocols, and is therefore generalizable to other scenarios. Our resulting hybrid engine is faster than either purely logical or purely numeric alternatives, making it potentially feasible for interactive production use.

%% file: intro.tex
\section{Introduction}

Network management is inherently complex. Operators contend with a bewildering array of technologies, whose interactions are difficult to predict, and struggle to express their diverse configuration goals in the setup of a live network. The practice of network management could be improved if operators were equipped with better tools for understanding and resolving trade-offs in network configuration.

We propose a formal approach for managing performance trade-offs in complex networks. While the use of  formal specification is not new, several past efforts have been hampered by excessive devotion to a single algorithm or technique; in contrast, our methodology is centered around a common representation to which different algorithms can be applied. 

The problem we study is how to enforce network invariants about traffic flow assertions provided by the operator which ought to be true, but might fail due to operational circumstances. These invariants may be logical, concerning the desired outcome of the network routing process and the paths which are selected, or numeric, relating to the aggregate flow of traffic. While there are many existing approaches to the numeric traffic management goals, based on techniques from operational research, and there are tools that tackle the logical combinatorial goal, there is no unified method. Therefore, key tasks in the practical management of networks are still difficult: operators cannot easily understand their available options for configuration, assess trade-offs between different goals, or perform joint optimization.

This paper presents a focused effort at one particular problem within this general space, namely, we look at the problem of reconciling the enforcement of MED values in BGP routing with individual AS policies without hurting traffic engineering optimizations within that AS. 

Although we present FixRoute in the context of solving MED induced oscillation problems, it can be used in a much more general framework. Namely, the  tool accepts a set of logical preferences (MED values in our problem) and numeric, convex, objective functions (TE cost) as its optimization objectives and produces a set of link weights that results in routing decisions that satisfy both objectives. For example, the operator may want all traffic coming in from source $s$ to go through a particular middlebox $m$ for processing before being forwarded to the destination. FixRoute can translate this rule to a set of logical preferences that favors all paths from $s$ that include $m$ to all other paths and allocate link weights so as to cause minimum disruption to TE objectives. As an additional example the methods presented in this paper can be used to implement the optimizations proposed by~\cite{feamster} which reroute flows in the network to minimize the aggregate peering and infrastructure cost without exceeding network capacity. 

While it is true that MED values are mostly ignored in today's BGP route selection policies, we believe that one of the reasons for this is the complexity of configuring such policies to maintain autonomous level (AS) level objectives while allowing for possible preferences from neighboring ASes to be accommodated as well. Given that the intersection of these two objectives may lead to the occurrence of oscillations and instabilities in the network, the tendency to ignore MED values is understandable. 

However, we believe that MED values are useful in their own right. Today, many network administrators use alternative means, e.g., AS prepending, to try to enforce certain preferences in BGP route selection. These methods reduce transparency between participating ASes and make it harder for individual ASes to enforce their own ingress routing preferences.  In this paper, we use the MED value oscillation problem to illustrate the viability of our approach in solving network policy problems.   

Our tool FixRoute provides and implements a structure to ensure that logical invariants are enforced and numeric traffic optimality goals are reached as far as possible; section \ref{sec:probstatement} explains the specific problem in more detail.

FixRoute provides a means for assessing trade-offs between different objectives in network management. Previous attempts at automating management tasks have sought to incorporate many objectives into a fixed strategy.  The problem with this approach is that it does not match the typical modes by which networks are actually managed, where the trade-offs are contextual, and it may be difficult to express a single strategy that encompasses all needs.

It is not feasible to ask a network operator to use a system where the first step is to write a complex optimization formula, even if we can then solve it exactly. In contrast, our approach enables operators to experiment with the interplay of these objectives, within a sound conceptual design underpinned by a proven mathematical formalism.

%% file: probstatement.tex
\section{Problem Statement}
\label{sec:probstatement}
Our developed configuration repair tool chain, FixRoute, is for repairing broken invariants within a single AS. We will assume that network-internal paths are established according to a shortest path rule, as is standard. External paths will be selected according to the process of the Border Gateway Protocol (BGP). Note that desired outcomes for BGP policy can be expressed in terms of given path preferences: assertions about which paths should be preferred as a result of running the distributed decision process. 

It has been known that BGP routing can fail through persistent oscillation~\cite{first},~\cite{second}, where the protocol will switch between best paths at high speed, causing increased control plane traffic, router overhead, and degradation of data plane service. Within a single AS, oscillations are typically caused by disrupting interactions with the policies of the neighboring ASes. One important class of convergence failure that FixRoute aims to handle is MED-induced oscillations ~\cite{second}. 

Proposed solutions to persistent oscillations today are non-ideal. They range from waiting for misconfigurations to go away, ad-hoc reconfigurations,  to design guidelines ~\cite{second} that provide heuristics on reducing particular classes of oscillations, and perhaps most drastic of all, completely ignoring the MED values. For mitigation, techniques such as route flap damping ~\cite{third} can lessen the effects of oscillations, without actually tackling the root cause. This could be helpful in allowing extra time for ad-hoc reconfigurations, but it is by no means a solution. In the long term, redesigning the protocols involved (as in ~\cite{fourth}, for example) would be the best option, but that future has not yet arrived.

None of the existing solutions take traffic engineering (TE) into consideration. While both convergence and TE relate to the action and outcomes of the routing protocol, they have quite a different flavor. TE is concerned with numeric constraints on the flow of data across network links, and is traditionally approached using the classic numeric techniques of operational research and combinatorial optimization. In contrast, the question of convergence requires reasoning about the action of distributed protocol, and the techniques here are drawn from logical reasoning. These goals are both important, but are in conflict because there is no particular reason why a TE-optimal configuration should satisfy BGP convergence goals, nor why the attempt to propitiate BGP should result in a configuration with good traffic flow properties. Ideally, one would want to have a solution concept that would allow the network operator to assign preference as to which constraint should be of greater importance.

Our goal is to take a network, and mandated preferences for path selection, and synthesize a sequence of configuration updates that (1) upholds the given preferences, and (2) does not cause severe harm to TE cost. In particular, we address the case where the mandated preferences are chosen to avoid an inter-AS BGP oscillation problem. The configuration parameters we target are the link weights of the interior routing protocol.

We focus on link weight changes because these are frequently used for TE, including by automated processes, so there is precedent for deciding what the link weights should be in order to meet a network design objective. The fewer actions needed to be taken, the faster the recovery process will be. The computation of the sequence should be efficient and reliable. It must be able to produce results that are at least as good as could be reached manually. These characteristics would make it suitable for hands-off deployment, automatically reconfiguring link weights as problems arise.

\subsection{MED Induced Oscillations}

\begin{figure}[h!]
  \centering
    \includegraphics[width=0.5\textwidth, height=5cm,keepaspectratio]{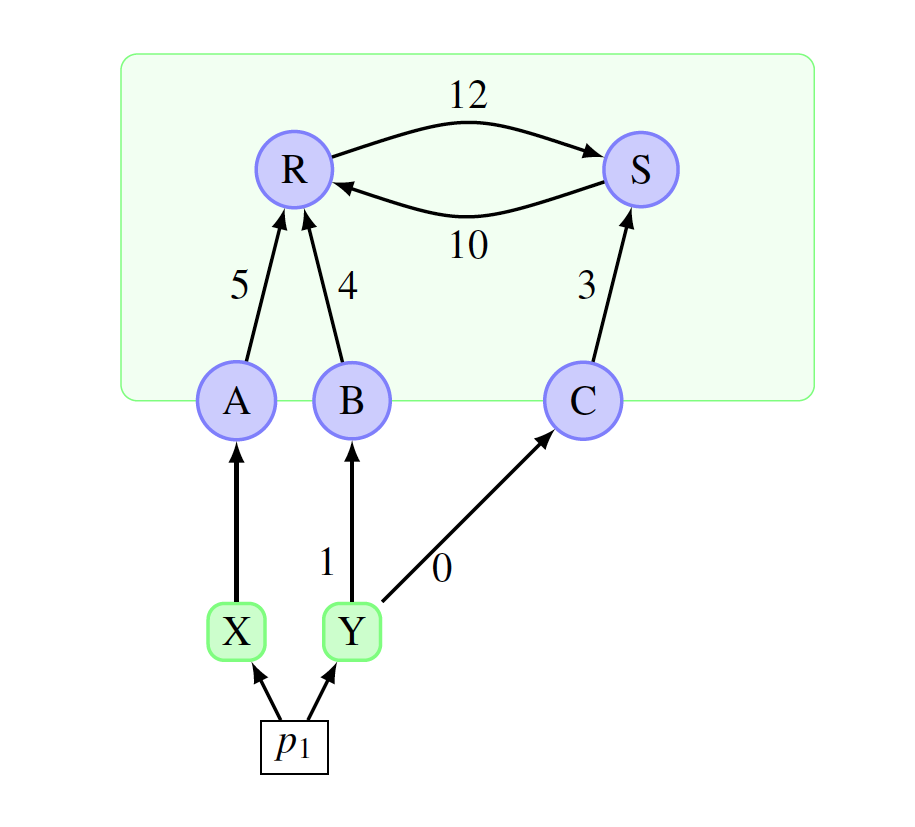}
      \caption{\label{fig:MED} Minimal network configuration with MED oscillation}
\end{figure}

See Figure ~\ref{fig:MED} for the topology. The main AS has five routers configured into two route-reflector clusters: A and B are clients of R, and C is a client of S. The three boarder routers are connected to two external ASes (shown in green). The  annotations on the AS-internal links indicate IGP distance values, and those on the external links are MED values.

The cause of the oscillation is the multiple-exit discriminator (MED) BGP attribute. MED is taken into account when choosing between routes with the same next-hop AS; for these, routes with lower MED values are better. Ties are broken by further attribute comparisons, principally intra-AS path lengths, which are computed by a separate shortest-path protocol. Therefore, for router R here, the route through B is better than the A route (MED) and the A route is better than the C route (shortest paths). 

Cycles of preferences, as seen in the previous example, can cause oscillations, and changing the weight values can create/destroy a cycle. However, due to the number of paths and complexity of the route computation it is not easy to judge which changes should be made in practice. The general principle is that the IGP path weights should be consistent with the preferences imposed by the MED attribute. In this way, preference cycles like those illustrated above are repressed. This can be achieved by reweighing the IGP links, by suppressing or overwriting MED values, or both. We attempt to change the link weights so that they become consistent with the given MEDs. That having been said, we recognize that it is important to provide operators with insights into how `binding' the MED constraints are, so that they can be better informed about the consequences of accepting these values: we propose that by giving our tool different combinations of the possible incoming MEDs, users could gain insight into the tradeoffs involved. Variations of the unequal split optimization (which we will discuss in section ~\ref{sec:unequal}) can be used as a lower bound on the minimum TE shift that the operator should expect.

\subsection{Link Weight Optimization}

Instead of an ad-hoc solution, FixRoute plans to repair an existing oscillation network instance, as a sequence of link weight updates. In a given AS, routers use assigned link weights to run shortest path algorithms to compute path to given destinations. Adjusting link weights allows the network operator to simultaneously enforce route preferences while meeting traffic engineering costs.

%% file: overview.tex
\section{Overview}
\label{sec:overview}

\begin{figure}[h!]
  \centering
    \includegraphics[width=0.5\textwidth,height=4cm,keepaspectratio]{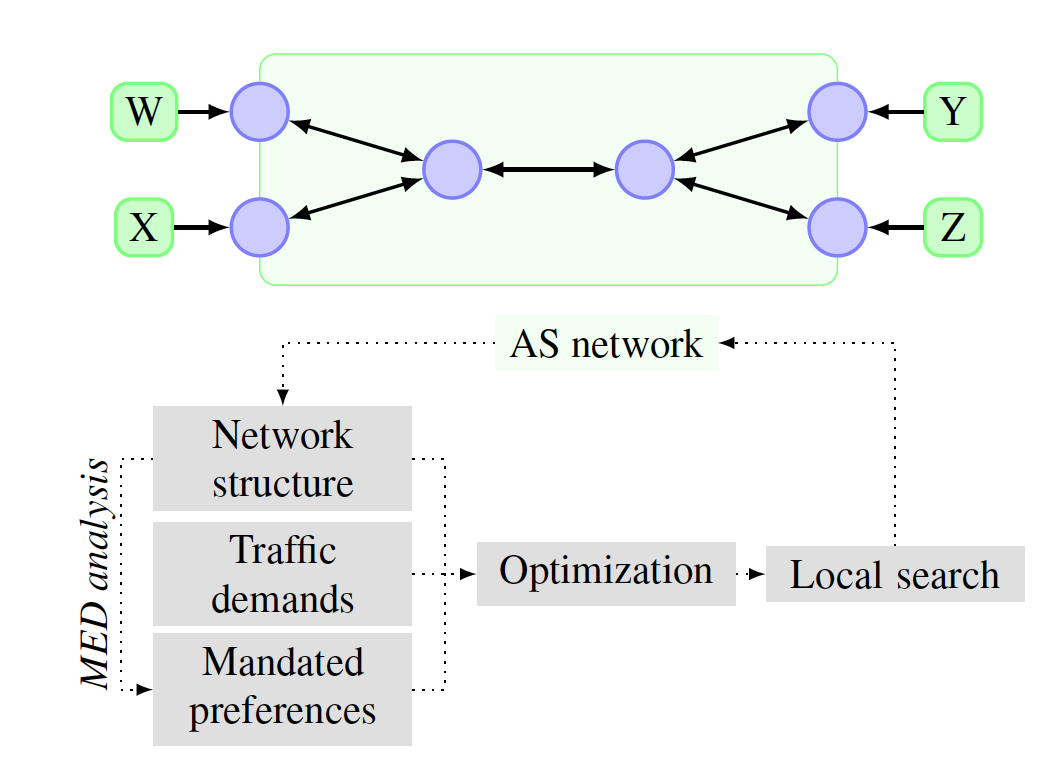}
      \caption{\label{fig:overview} Overview of the component parts of FixRoute.}
\end{figure}

Figure ~\ref{fig:overview} shows a high-level overview of various components of FixRoute. FixRoute is designed for configuring the link weights within an AS, in order to meet two goals: (1) implementing route preferences specified by the AS's network operator, and (2) optimizing for TE based metrics. 

FixRoute takes as input the following information:

\begin{itemize}
\item \textbf{Network instance. } Includes the topology of the entire AS's internal network, and existing link weights. Weights may be specified as unknown, to be set by FixRoute. 

\item \textbf{Traffic demand matrix. } For each source $s$ and destination $d$, specifies the amount of data from $s$ to $d$.

\item \textbf{Mandated preferences. } A listing of which paths are to be made `better' than which other paths. These can also be derived via analysis from the network instance and information about possible incoming routes.

\end{itemize}

In addition the system takes an optimality threshold parameter $\gamma$. This will be the tolerance set by the network operator specifying the allowed deviation of the solution's TE cost from the original. FixRoute then attempts to generate a sequence of changes to existing/unspecified link weights. Within FixRoute, we propose the following optimization strategies, proven correct by utilizing the PSPP concept (Section ~\ref{sec:PSPP}).

\begin{itemize}
\item \textbf{Max-SMT strawman.} (Section \ref{sec:SMT}) We express the input network instance as a PSPP, and then utilize a Max-SMT solver to enumerate all possible paths in the network to find an optimal solution. This solution requires time complexity that is exponential in the number of nodes.

\item \textbf{Unequal splitting. }(Section \ref{sec:unequal}) Assuming routers permit unequal splitting of traffic among paths of equal lengths, we utilize a two level optimization strategy that generates the set of link weight changes. The solution yields the best possible tradeoff between achieving a TE optimal solution and minimizing the number of weight changes.

\item \textbf{Equal splitting. }(Section \ref{sec:equal}) If routers require equal splitting of traffic among paths of equal lengths, finding the optimal solution is NP hard. we propose a modification to an existing local search heuristic that comes reasonably close to the TE optimal ~\cite{fifth}. 

\end{itemize}

%% file: PSPP.tex
\section{Path Preference Model}
\label{sec:PSPP}

Routing policy is known to be approachable in terms of various mathematical formalisms, including algebraic and combinatorial models ~\cite{sixth},~\cite{seventh}, ~\cite{eigth}. In past work, these have been used to study routing protocol convergence, and to explore the expressive power of different policy schemes. However, we will use a policy model to support reasoning about possible outcomes of the routing process: which paths might be selected as a result of procotol execution, given a network configuration. The formalism is based on partially specified configuration information, whereby `gaps' in the configuration can be filled in, according to some high-level goal. 

The partial stable paths problem (PSPP) is a combinatorial representation of partially-specified routing protocol convergence ~\cite{sixth}, ~\cite{seventh}, ~\cite{eigth}. A PSPP instance consists of an enumeration of possible network paths, and the known preferences among them. In our case, it does not include preferences induced by link weights, but specifically does include preferences that are mandated by the user. The missing, undecided, preferences are the `gaps'. An instance can be completed by filling in these gaps, therefore determining the relative preference of all paths, and the values of the configuration parameters that give rise to those preferences. Any potential oscillation problems can be detected as `conflicts' of path preference in this structure; specifically, the presence of a cycle of preference, where each path is considered to be strictly preferred to the next ~\cite{ninth}.

Absence of a preference cycle is equivalent to the key properties of safety and robustness of a routing configuration ~\cite{seventh}, ~\cite{eigth}, ~\cite{ninth}. Safety means that all fair protocol executions, from any given configuration state, will eventually converge to the same unique fixed point. This means that there can be no permanent oscillation in the protocol, and furthermore that the protocol produces the same result regardless of timing. This second feature is held to be important because multiple stable solutions typically occur accidentally-there is some `intended' final state, but in some circumstances the protocol could end up wedged in an unusual alternative routing state ~\cite{tenth}. Robustness says that the state will remain safe even if some nodes and arcs are removed. This is meant to rule out configurations with certain latent problems, where an oscillation is being hidden by some other aspect of the routing state, and where the failure of the part of the network will cause the oscillation to reveal itself. Note that these properties hold regardless of the dynamic execution of the protocol. 

\begin{definition}{\label{def:PSPP} A partial stable paths problem (PSPP) consists of a directed graph $\mathcal{G}(V,E)$; a designated destination node $d \in V$; and for each $v \in V$:
\begin{enumerate}
\item A subset $\mathcal{P}_v$ of the simple paths from $v$ to $d$, and
\item A partial order $\le_v$ on $\mathcal{P}_v$.
\end{enumerate}
The paths in $\mathcal{P}_v$ are the permitted paths and the order $\le_v$ gives path preferences.
}
\end{definition}
Because the order $\le_v$ is partial rather than total, this definition allows for some preferences to be undefined. These preferences are given abstractly, but in real protocols they are induced by attributes of the paths and by fixed rules in the routers.

\begin{definition}{ One PSPP $\mathcal{P}$ is a refinement of another PSPP $\mathcal{Q}$ if they have the same graph, same destination node, and same sets of permitted paths; and for each node $v$ if $p \le_v q$ for some paths $p$ and $q$ according to the order in $\mathcal{Q}$, then $p \le_v q$ according to the order in $\mathcal{P}$ as well.
}
\end{definition}
A special case of refinement is when all path preferences have been set one way or another, and so the orders $\le_v$ are all total. The concept of preference cycle is defined with reference to the path digraph. This data structure contains information about paths and their dependencies.

\begin{definition}{\label{def:di} Given a PSPP $\mathcal{P}$, its path digraph is the directed graph $(V', E')$ where the nodes in $V'$ are the paths in $\mathcal{P}$, and $(p,q)$ is an arc in $E'$ if  and only if either:
\begin{enumerate}
\item $p$ is a suffix of $q$, or
\item $p$ and $q$ have the same source node $v$, and $p \le_v q$ but $\neg (q \le_v p)$.
\end{enumerate}
For these two cases, we use the notation $p \rightarrow q$ and $p \preceq q$ respectively.
}
\end{definition}

A path is said to `depend on' its predecessors in this digraph, because it is the presence/absence of these paths which completely determine whether the path should be part of the routing solution. For a path to be selected, its suffixes must be present (corresponding to the BGP rule that one is only allowed to use routes which have been advertised and not withdrawn by neighbors), and no better path can be present. A preference cycle is simply a cycle in this digraph. 

The PSPP model is well-adapted for reasoning about configuration repair. We can ensure that certain desired route preferences will indeed hold, and then use that information to deduce information about other paths and the underlying configuration. In particular, oscillation can be avoided by requiring that the structure be acyclic, or linearizable: this means that it is possible to find a (single, global) ordering of all the paths, which is compatible with both the suffix relation and the preference $\le_v$ relations. A linearization may be witnessed by giving a numeric `rank' to each possible path, where each path must have strictly larger rank than any of its predecessors.

%% file: repair.tex
\section{Repair With Optimization}
This section describes our method for carrying out invariant repair, while trying not to harm the TE optimum. We first present an inefficient strawman solution that linearizes a PSPP instance, followed by two numeric-based techniques that leverage the PSPP model for correctness, as applied to routers that permit unequal and equal traffic splitting respectively for paths of the same length. We assume that preferences for the linearization are those which arise from oscillation-avoidance; these can be derived in a straightforward way from the network configuration and information about incoming paths.

\subsection{Oscillation Avoidance}
As discussed above, the oscillation problem arises from non-linearizability of route preferences: in short, when the preferences imposed by MED attributes (or by the administrator) don't match those induced by the link weights. Therefore, if we take the incoming MEDs as binding, the desired invariant for our link weight is that they should result in the same preferences as the MEDs do; whenever $p$ and $q$ are two paths, where $p$ is preferred to $q$ by reason of MED, then $p$ should also be of lesser weight than $q$. 

It is conceptually straightforward, given an account of the incoming routes, to derive these preferences. The paths in the network can be enumerated, and for each pair, we can determine whether MED will be the decisive attribute in comparison. If so, then this should be one of the mandated preferences. On the network of Figure ~\ref{fig:network}, this can be done in less than a second by a naive algorithm.

\begin{figure}[h!]
  \centering
    \includegraphics[width=0.5\textwidth]{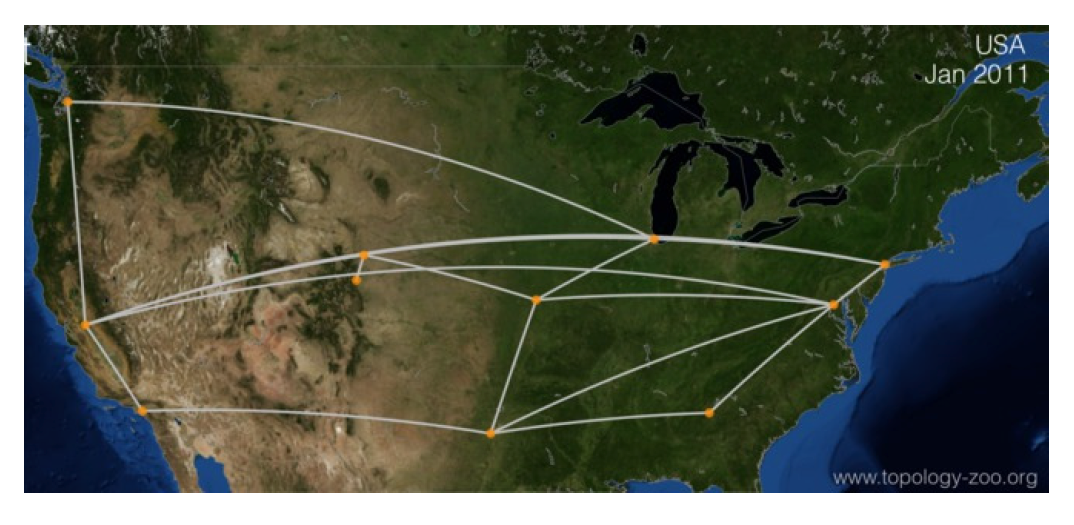}
      \caption{\label{fig:network} AS1239 topology, from Topology Zoo ~\cite{fourteenth}, ~\cite{fifteenth}.}
\end{figure}

\subsection{Strawman: Pure Linearization}
\label{sec:SMT}

This technique uses Max-SMT, which is the SMT version of Max-SAT. That is, SAT problems are about satisfiability of Boolean formulae; SMT (satisfiability modulo theories) adds a richer language, including integer arithmetic; and `Max' means that the goal is to satisfy as many clauses as possible, not just all-or-nothing. Each claus in the Max-SMT program is associated with a numeric cost, and we can then ask a solver for a model which maximizes the total cost of the assertions that are satisfied. Practical Max-SMT solvers may fall back to heuristic rather than exact methods, if the number of variables or complexity of the formula are high enough. Assertions that are unweighted are taken to be constraints that must always hold. The Yices solver ~\cite{eleventh} is one that is capable of operating in Max-SMT mode.

We generate a Max-SMT program for the repair problem using the following steps: 
\begin{enumerate}
\item Enumerate the paths in the network and give each one a numeric ranking variable. 
\item Assert the path digraph constraints ($\rightarrow, \preceq$) on the rankings, arising from the network structure and the mandated preferences.
\item Associate each path with its sum of link weights, and assert that if a path has a lower ranking than another, then it must have a lesser sum of weights.
\item Assert (with finite weight) that each individual link weight should have its original value.
\end{enumerate}
Then, a solution to this program will represent a linearization of the path digraph, with associated link weights such that the path preferences are indeed induced by those weights. Moreover, because of the maximization of satisfiability, as many link weights as possible will be unchanged.

This approach fails to achieve the desired TE properties and is rather inefficient. In contrast, the next section presents a full method, also based on the PSPP model, that does achieve all of the problem goals.

\subsection{Numeric Optimizations}

We now present our numeric optimization technique, which
incorporates the constraints of the desired invariant, and
therefore tries to produce a solution that is optimal subject to
those constraints. The first solution involves an optimization problem which provides a set of link weights that satisfy the given mandated preferences, without harming TE cost too much. This formulation assumes that when path of equal best cost exist, the traffic flow can be split between them in any desired proportion. Therefore, when the method produces an `optimal' weight assignment, the assessment of its cost (and even its feasibility with respect to bandwidth constraints) may not match the true cost in a system where such splitting is not possible. This is typical for the current IP data plane, where if traffic splitting is allowed, it may only be possible to split in equal ratios (e.g. ECMP ~\cite{eq}).

The case where no traffic splitting is permitted can be handled by a small adjustment to the linear programing objectives. For the more interesting case of equal splitting, we show a heuristic local search procedure that takes the linear programming solution as its starting point, and adapts it to find a `nearby' link weight assignment which meets the refined objectives. In the next sections we will investigate the practical performance of the combined procedure, and show that it is more effective than running the heuristic local search alone \footnote{ The reason is that the bulk of the effort of finding the solution still rests with the linear program which finds the best possible solution, the heuristic then tries to find the solution that is as close as possible to the linear optimization configuration.}. 

In the following, we represent the network topology as an undirected graph $G(N,E)$, where $N$ is the set of nodes and $E$ the set of edges. A simple path in this graph is a sequence of connected nodes in which every node in $N$ appears at most once. The notation $[ i_1,i_2,\dots,d]$ stands for a path that goes through nodes $i_1, i_2, \dots$ and has destination $d$. An empty path for node $d$ is denoted as $\epsilon_d$.

The set of all paths to destination $d$ is defined as $\mathcal{P}_d$. A set of mandated preferences for paths with destination $d$ is given by the binary relation $\mathcal{S}_d$ on $\mathcal{P}_d$, where $(p,q) \in \mathcal{S}_d$ means that path $p$ is to be preferred to path $q$. Per node preferences $p \le_v q$ can be defined as $(p,q) \in \mathcal{S}_d$. 

Recall from Definition \ref{def:di} that the path digraph $\mathcal{G}'$ for this given instance has the paths in $\mathcal{P}_d$
 as its nodes (for each destination $d$), connected by edges given by the preferences in $\preceq$, derived from $\le_v$, and the immeidate suffix relation $\rightarrow$ on paths. Note that for every destination node $d$ in $G$, and every path digraph node $p=[i,\dots,d]$ in $\mathcal{G}'$, there is a chain in $\mathcal{G}'$ connecting $\epsilon_d$ to $p$, consisting of all the suffixes of $p$.

Our problem is to assign link weights to all the links in $E$, so that the induced path weights will be consistent with the mandated preferences in each $\mathcal{S}_d$. These mandated preferences may be given in order to avoid oscillation, or for other purposes. As previously observed, avoidance of oscillation corresponds to finding a linearization of $\mathcal{G}'$ and a witness to this linearization is an assignment of a numric `ranking' to each path, in a way consistent with the digraph structure (i.e., if path $p$ is reachable from $q$, then it should have a strictly larger ranking). This notion can now excercise a double duty: rather than being arbitrary numbers, these rankings (potentials), will now be precisely the sum of the weights for each path. Let $\lambda_p$ stand for the potential assigned to path $p$. If $p$ is a suffix of $q$ then it should have a smaller potential, which remails true when the potentials are sums of link weights: $q$ contains all the links that $p$ does, so its potential should be at least as large as that of $p$ \footnote{Note that we are assuming positive link weights.}. Similarly $(p,q)$ in $\mathcal{S}_d$ tells us that path $p$ should have lower weight than path $q$. Of course it is not enough to simply linearize $\mathcal{G'}$, since we need to make sure not only that the potentials are consistent with the digraph structure but also that they are realizable in terms of link weights. In addition to the desired link weights $w_{ij}$ for $(i,j) \in E$, we generate virtual link weights $w'_{pq}$ for each mandated preference $(q,p) \in \mathcal{S}_d$. This results in the following description of the feasible set of link weights:
\begin{equation}
\begin{aligned}
\label{eq:feasibility}
&\lambda_p -\lambda_q=w_{ij} \quad (i,j) \in E, q=[j,\dots,d], p=(ij)q\\
&\lambda_p -\lambda_q=w'_{pq} \quad (p,q) \in \mathcal{S}_d\\
& \lambda_{\epsilon_d}=0\\
& w'_{pq} =w_{ij} \quad (i,j) \in E, q=[j,\dots,d], p=(ij)q\\
& 0 \le w_{ij} \le w_{max} \\
& 1 \le w'_{pq} \le w_{max} \\
& 0 \le \lambda_{p} \le |N|w_{max}, \quad p \in \mathcal{P}_d\\
& w_{ij}, w'_{ij} \text{integers}.
\end{aligned}
\end{equation}

The value $w_{max}$ is the maximum permitted link weight. Checking \ref{eq:feasibility} requires enumerating all paths in $\mathcal{G}$. However, the followoing statement allows us to reduce the search space.

\begin{lemma}
Let the set of all paths that appear in $\mathcal{S}_d$, together with all of their suffixes be $\mathcal{S}'_d$. Any solution for the restriction of \eqref{eq:feasibility} to paths in $\mathcal{S}'_d$ determines a solution for the unrestricted problem.
\end{lemma}

\begin{proof}
Let $\lambda, w,$ and $w'$ give a solution for the restricted problem. For any path $p$ in $\mathcal{P}_d \setminus \mathcal{S}'_d$, set $\lambda_p=\sum_{(i,j)\in p} w_{ij}$. We show that this yields a feasible solution to the unrestricted problem. 

Let $p$ and $q$ be paths in $\mathcal{P}_d \setminus \mathcal{S}_d$. If neither is a suffix of the other, then their assigned potentials are satisfactory. For in this case, they are unconnected in $\mathcal{G}'_d$, and so neither $\lambda_p \le \lambda_q$, nor $\lambda_q \le \lambda_p$ is required. Otherwise, if (without loss of generality) $p$ is a suffix of $q$ there is a path in $E$ that extends $p$ to yield $q$, and $\lambda_q-\lambda_p$ gives the total weight of this path. This is consistent with the first condition of \eqref{eq:feasibility}. The same argument carries through for path $p$ and $q$ such that $p \in \mathcal{S}'_d$ and $q \in \mathcal{P}_d\setminus \mathcal{S}'_d$ with the caveat that now the only possibility is for $p$ to be a suffix of $q$. The remaining conditions relate to paths in $\mathcal{S}_d$, and are already satisfied.
\end{proof}

Thus, as long as any set of weights satisfies ~\eqref{eq:feasibility} for all paths in $\mathcal{S}'_d$, it can be considered safe. Note that if $p,q \in \mathcal{S}'_d$ and $q$ is a suffix of $p$, then condition ~\eqref{eq:feasibility} needs to be checked to show that there exists a set of positive integer weights that result in an assignment of potentials satisfying the mandated preferences and $\lambda_p \ge \lambda_q$. We denote the initial unsafe weights as $w^e$. In order to perform repair with minimum weight changes, the cost function $h(x)$ is defined as

\begin{align}
\label{eq:cost}
 h(x) = \left\{ 
  \begin{array}{l l}
    x^2 & \quad \text{if $|x| <\epsilon$}\\
   M & \quad \text{otherwise,}
  \end{array} \right.
  \end{align}
  in which $M$ is a large number. Thus we arrive at the following formulation of the problem:
 
 \begin{equation}
 \begin{aligned}
 \label{eq:opt1}
 \min_{p \in \mathcal{S'}, w_{ij}, w'_{ij}} \sum_{i,j} h(w_{ij}-w^e_{ij})\\
 s.t. \quad \text{conditions in equation ~\eqref{eq:feasibility}}.
 \end{aligned}
 \end{equation}
 
 Using the method derived in ~\cite{twelvth}, this problem can be solved by formulating it as a dual network flow problem in $O(nm\log{n^2/m})\log(nU)$ time, where $n=|\mathcal{S}'_d|$, $m=|w|+|w'|$, and $U=|N|w_{max}$.
 
 Though \eqref{eq:opt1} allows for repair with minimum weight changes in polynomial time, it completely ignores the impact on network traffic. In the following subsections, we provide means for network operators to perform repair while remaining within $\gamma$ percent of the TE optimum.
 
 \subsection{When routers can perform unequal splitting}
 \label{sec:unequal}
 Traditional link state routing protocols such as OSPF split traffic equally among paths of equal length. the previous work of Fortz and Thorup ~\cite{fifth} has shown that solving the TE problem with equal splitting is NP-hard. Though we address this problem by providing a heuristic in the next subsection, it is important to note that new routing policies such as MPLS have made it possible to realize optimum solutions that require unequal splits in the network. A solution to this more general problem also makes it possible to gauge whether an oscillation free solution with minimum traffic disruptions is feasible. This allows network operators to identify whether they benefit from incorporating other means of eliminating oscillations in their networks. For these reasons, we provide a brief summary of our solution, and then focus on the more prevalent problem of equal splits.
 
 We denote the flow going through each link $(i,j) \in E$ as $x_{ij}$, and the link incident matrix of the graph $G$ as $A$. The demand matrix is captured using the vector $b$. The flow $x$ is thus the solution to the following optimization problem:
 
 \begin{equation}
 \begin{aligned}
 \label{eq:netflow}
 &\min_x w^T x\\
 &s.t. \quad Ax=b\\
 &x\ge 0
 \end{aligned}
 \end{equation}

Defining the cost function of the TE problem as $\sum_{i,j} \phi(x_{ij})$ and the previous load on link $(i,j)$ as $x^e_{ij}$ the problem can be formulated as:

 \begin{equation}
 \begin{aligned}
 \label{eq:game}
 &\min_{w,x} \sum_{i,j} h(w_{ij}-w^e_{ij})+M'\sum_{i,j}(\phi(x_{ij})-\phi(x^e_{ij})(1+\frac{\gamma}{100}))\\
 &s.t. \quad conditions \eqref{eq:feasibility}, \eqref{eq:netflow},
 \end{aligned}
 \end{equation}
 
where $M'$ is a large number. By changing the values of $M$ and $M'$, an operator can specify a Pareto frontier of his optimization goals and indicate which goal should be of more importance. $\phi(x)$ may be any convex function of $x$, however,  we take $\phi(x)$ to be the cost function introduced in ~\cite{fifth}.

 Equation \eqref{eq:game} is a non-convex two level optimization problem also known as a Stackelberg game \footnote{To see the analogy with a Stacklberg game, note that the network administrator acts as the master that can set the weights ($w$) for the network and has a TE cost function ($f(x,w)$) that is dependent on the network flow ($x$) which in turn is realized by the slave problem \eqref{eq:netflow} (the network).}. Two level optimizations have been extensively studied in the literature and numerous solutions have been proposed. In particular we found that the solution proposed in ~\cite{thirteenth} provides a means of efficiently solving \eqref{eq:game}. As the main role of weights in our formulation is to assign an ordering among path, the solution to the relaxation of this integer programming problem often provides an exact integer assignment that results in the same ordering.  For lack of space the reader is referred to ~\cite{thirteenth} for more detail on how to solve \eqref{eq:game}.

 \subsection{Solving the problem for equal splits}
 \label{sec:equal}
 The solution of \eqref{eq:game} yields the best possible cost tradeoff between achieving a TE optimal solution and minimizing the number of weight changes, in providing an oscillation free weight setting. However, most routers in the internet don't allow arbitrary splitting between equal cost paths.
 
 The authors of ~\cite{fifth} showed that finding a solution to the TE problem with equal splitting is NP hard; adding our mandated preference constraints to this formulation only exacerbates the complexity. They provide a heuristic algorithm performing local search operations that come reasonably close to the TE optimum. In this section we introduce our modification to this approach that ensures that the resulting weights satisfy ~\eqref{eq:feasibility}. We do not explain all the nuances of their procedure here, but concentrate on the general outline and our changes.

 The local search procedure attempts to find an optimal set of link weights by taking an initial weight assignment, exploring a sample of its `neighborhood' of similar assignments, and repeating the search on the best assignment found. A neighborhood of a weight assignment consists of all the assignments which differ from it in one link weight only; the algorithm adaptively explores a different proportion of the neighborhood in each step, depending on the progress of the search so far. The cost function used sums the loads of all the links, as scaled by a fixed piecewise linear function. While this procedure is by nature heuristic, it has been shown to perform well in practice. We make three modifications to this setup:
 
 \begin{enumerate}
 \item Weight assignments that do not fulfill our mandated preferences are not explored.
 \item The search may terminate early, if a state is found that is within $\gamma$ of the original cost.
 \item The starting point of the search is a solution to \eqref{eq:game}. 
 \end{enumerate}
 
 The fist change is implemented by providing the search procedure with a copy of the mandated preferences. Each candidate state may be checked against these for viability: if the weights do not result in the correct inequalities, then the state is rejected (marked as infeasible).

 The reason for the second change is that our overall goal is to find weights that are no more than a factor $\gamma$ worse than the original TE cost. Normally, the program would run for some fixed number of rounds, and finally report the best weight assignment ever encountered; but now it can terminate earlier, if a state is found that is within $\gamma$ of the target.

 Finally, the choice of starting point is justified by the fact that the bulk of the effort has already been borne by solving \eqref{eq:game}, and the heuristic step is typically only required to adjust the weights slightly. Coupled with the presence of the $\gamma$ factor, the required number of iterations is typically quite low. Moreover, we are able to run multiple searches in parallel to refine the final solution, i.e. we run serveral simultaneous searches with different starting points to achieve better coverage of the solution space. This means that a viable solution can be found more quickly (by choosing the first to complete out of all searches), or that a solution of lower cost can be found (by allowing all searches to proceed further). It can also be helpful in presenting a user with multiple options for how to proceed, in the case where factors that are not part of the existing optimization need to be taken into account.
 
 The generation of these multiple starting points can be achieved by altering \eqref{eq:game}. Once we have one solution, we can add a term that assigns high cost if that solution is found again; and so on to find more and more possibilities. In our implementation, we use a simple driver program to spawn multiple local search processes. The first one to finish `wins' and the others are terminated. In the next section, we show experimentally that the resulting solution will satisfy all mandated preferences, and will be within a $\gamma$ factor of the original cost (if its feasible).

%% file: evaluation.tex
\section{Evaluation}
\label{sec:eval}
In this section, we present an evaluation of FixRoute. The goals of our evaluation are to demonstrate that given an existing oscillating network, FixRoute is able to generate a weight matrix that satisfies all mandated preferences and does not harm TE costs too much.

\subsection{Experimental Setup}
In sections \ref{sec:osc} and \ref{sec:optanal} we use a network topology derived from the backbone topology of AS1239 as recorded in the Topology Zoo ~\cite{fourteenth} (Figure \ref{fig:network}). In this setup, there are a non-trivial number of nodes and links, but more importantly there are many cycles and alternative paths. In these experiments, we set up a single router for each node in the dataset. All routers are configured as BGP route reflectors, and also speak OSPF along the same links. We attached three independent instances of the MED problem pattern to randomly chosen nodes. That is, each instance consists of two external ASes, advertising some prefix, connected to three random routers in the main AS, and using MED values as in Figure \ref{fig:MED}. Three independent oscillation problems to avoid, each inducing its own constraints on the IGP link weight values. We confirmed by experiment that each of the three prefixes could experience an oscillation, for some choice of weights.

We also give some validation results on a minimal example of MED oscillation ~\cite{second}. This topology, shown in Figure ~\ref{fig:secMED}, is a version of that of Figure \ref{fig:MED} but with an additional mirrored copy of the problematic MED values, in order to add a conflict between satisfying the two potential oscillation problems. 

\begin{figure}[h!]
  \centering
    \includegraphics[width=0.5\textwidth, height=5cm,keepaspectratio]{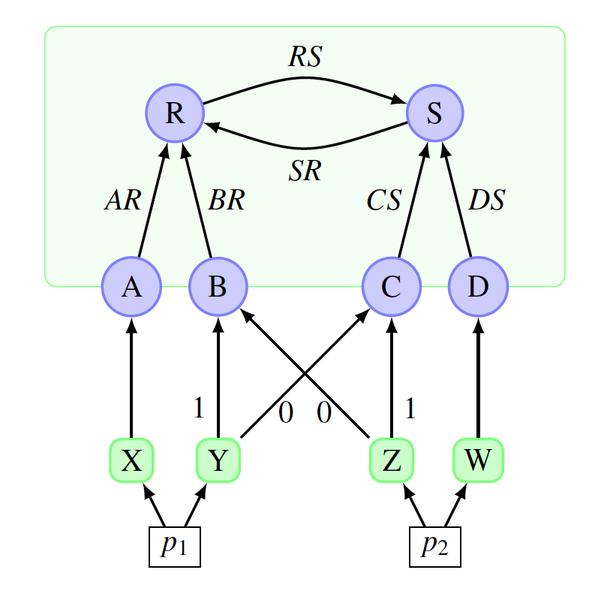}
      \caption{\label{fig:secMED} Network configuration with two problematic prefixes.}
\end{figure}

We set up each network topology on Emulab ~\cite{sixteenth}, using the Bird software router ~\cite{seventeenth} (v1.3.7). The Emulab machines were drawn from the 'pc850' pool, each of them being endowed with an $850$ MHz Pentium III processor and $512$ MB of PC133 ECC SDRAM. In the case of Figure \ref{fig:secMED}, six machines were used for the main AS nodes. Two additional machines were confugured for the external AS nodes (Bird BGP allows a router to present multiple AS personalities to its neighbors). A similar setup was used for the other example network, Figure \ref{fig:MED}. 

We made a single modification to the Bird code for the purposes of these experiments. Bird BGP contains an optimization to BGP best-route computation, which is intended to reduce the need for frequent recomputation based on the full routing table. This is based on the assumption that when the best route is $p$, and a new route $q$ arises (or an existing route has its attributes changed), then the only possible outcomes are for $p$ to remain, or for $q$ to be adopted, as the new best  route. That is, it assumes that BGP decision making satisfies the `independence of irrelevant alternatives'  axiom of decision theory ~\cite{eighteenth}. In fact,  this is not the case: in the example we are investigating, the presence of $q$ can cause the best route to change to a third route $r$, which is neither $p$ or $q$. We therefore removed this optimization. Given that we only have two external prefixes performance does not suffer.

As a side note, analysis of the predicate shows that the guideline in ~\cite{second}, that the inter-cluster links should be highly-weighted, is ineffective. The critical links are actually RB and SC. It is straightforward to generate examples where the inter-cluster links are huge but where the oscillation nevertheless continues. This observation is further evidence for why formal methods can be more helpful than rules of thumb. 

\subsection{Oscillation Analysis}
\label{sec:osc}
As our baseline, we first demonstrate that the the linearization in section ~\ref{sec:SMT} indeed results in the removal of oscillation. This is a necessary prelude to the more complex setting of sections ~\ref{sec:unequal}, and ~\ref{sec:equal}, where the linearization is embedded in the context of several other solution steps.

To assess the accuracy of the PSPP characterization, we carried out a random reweighing procedure: changing a link weight every thirty seconds (to a value between 1 and 100), for 23000 seconds and measured the control plane traffic. Figure ~\ref{fig:random} shows a representative graph of traffic covering a subsample of 10000 seconds. The periods of elevated traffic are due to routing oscillation. The shaded areas indicate when the network was unsafe according to our linearizability criterion.

\begin{figure}[h!]
  \centering
    \includegraphics[width=0.5\textwidth]{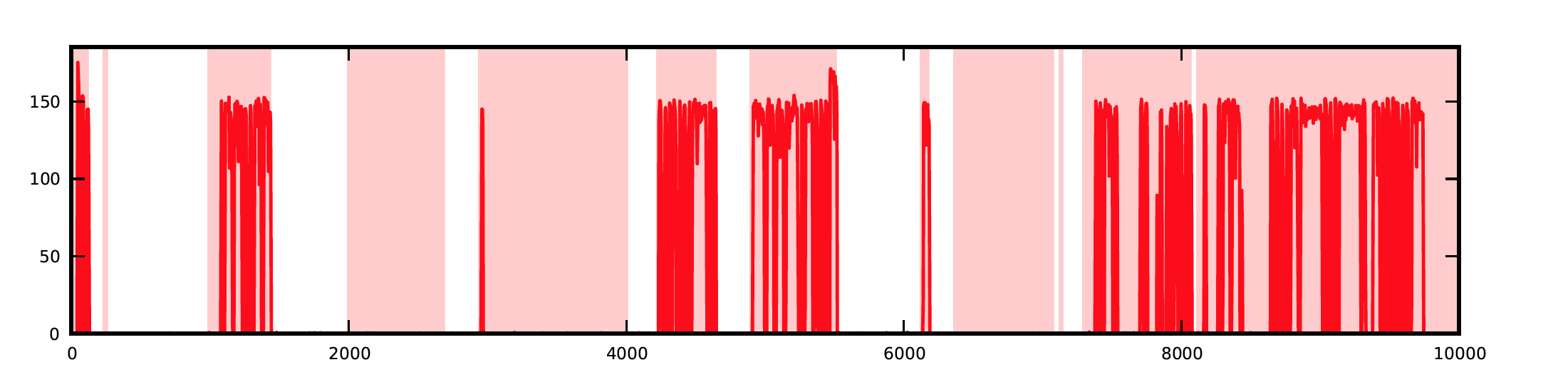}
      \caption{\label{fig:random} Control plane traffic in Kbps during the random reweighing process. The shaded area denotes states deemed to be potentially unsafe.}
\end{figure}

Oscillations occur in only $47.3\%$ of the states which were deemed unsafe; however, there was never an oscillation during a state which was considered to be safe. This matches our theoretical understanding of the PSPP model. 

For this run, the median oscillation duration was $166.5$ seconds, corresponding to $6$ link weights needing to be changed. The maximum was $1102$ seconds ($32$ changes), and the minimum was $3$ (1 change).

\begin{figure}[h!]
  \centering
    \includegraphics[width=0.5\textwidth]{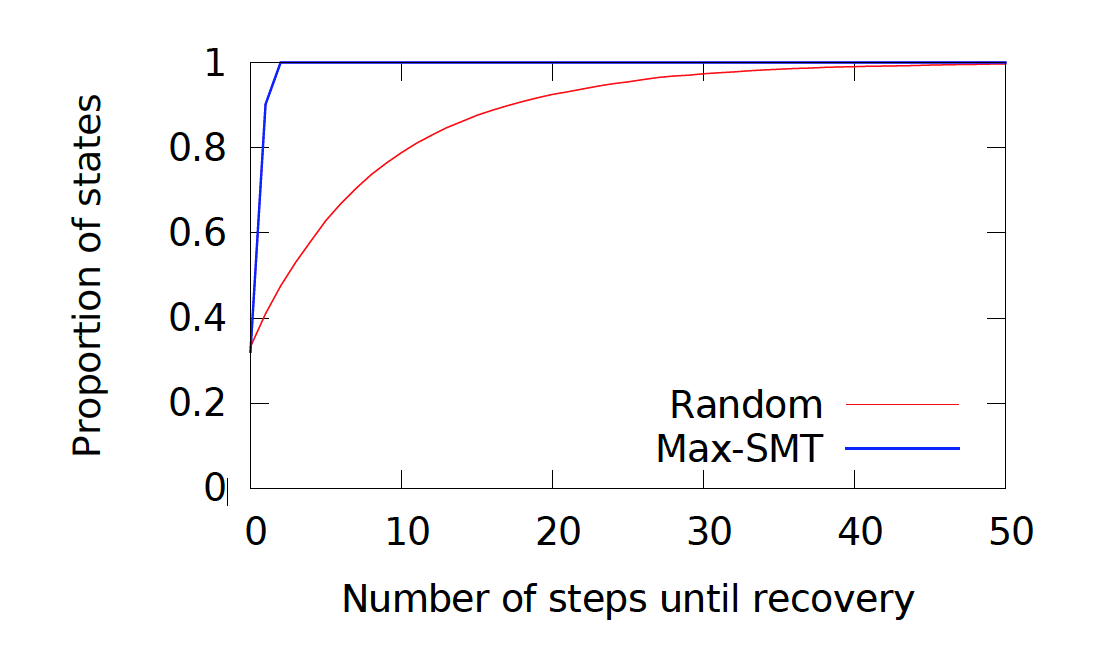}
      \caption{\label{fig:smtvsrandom} Number of steps until recovery.}
\end{figure}

Figure \ref{fig:smtvsrandom} shows the result of automated testing of the Max-SMT recovery procedure, compared to random recovery, on a sample of 10000 link weight assignments for Figure \ref{fig:secMED}. The random procedure changes link weights at random until the safety predicate returns true; as shown in the graph, dozens of reweighing steps may be required. 

On the other hand Max-SMT performs better: no state is more than two changes away from safety, and the bulk of them (about $90.2\%$) need only one change or are already safe, needing zero changes. Results were similar on the topology of Figure \ref{fig:network}.

\begin{figure}[h!]
  \centering
    \includegraphics[width=0.5\textwidth]{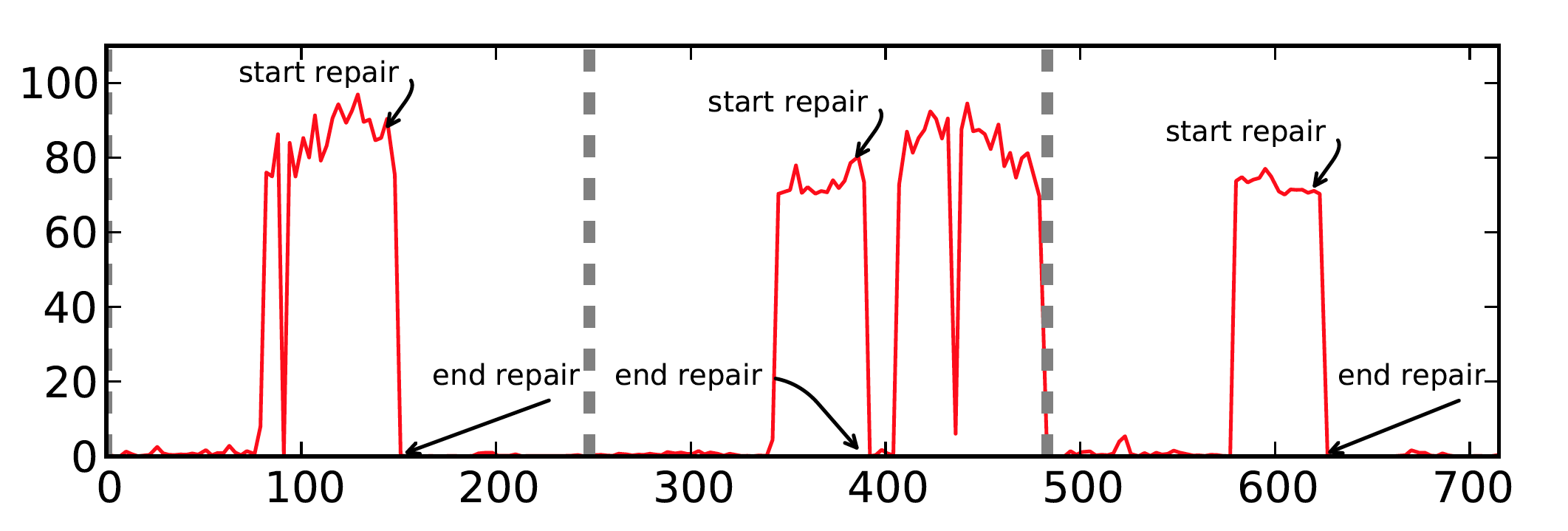}
      \caption{\label{fig:repairproc} Control plane traffic in Kbps during three attempts at solving oscillation problems.}
\end{figure}

Figure \ref{fig:repairproc} shows the observed control plane traffic for three experiments, using the same initial oscillatory state. The first and third of these use two different solutions emerging from our method. The second segment shows a failed attempt, where the link weight changes were chosen at random and the oscillation was disrupted but not removed. The arrows show the moments in time when each set of link weights began to be installed, and when the reconfiguration process had finished: at this point OSPF must reconverge, followed by BGP. While one would hope that practical networking wisdom performs better than a random process, it is not obvious that it can always achieve the optimal results here, which are orders of magnitude better. It also seems plausible that human ingenuity would normally require more than a few minutes to resolve these convergence failures. 

In the next section, we show the impact of these weight changes on network traffic, to see whether the cure is worse than the disease.

\subsection{Optimality Analysis}
\label{sec:optanal}
The simple Max-SMT does not preserve TE cost.  In further experiments with 100 demand matrices, the Max-SMT solution resulted in changes in link load of between $14-23\%$, with large quantities of traffic shifting from one path to another. What about FixRoute?

\begin{figure}[h!]
  \centering
    \includegraphics[width=0.5\textwidth]{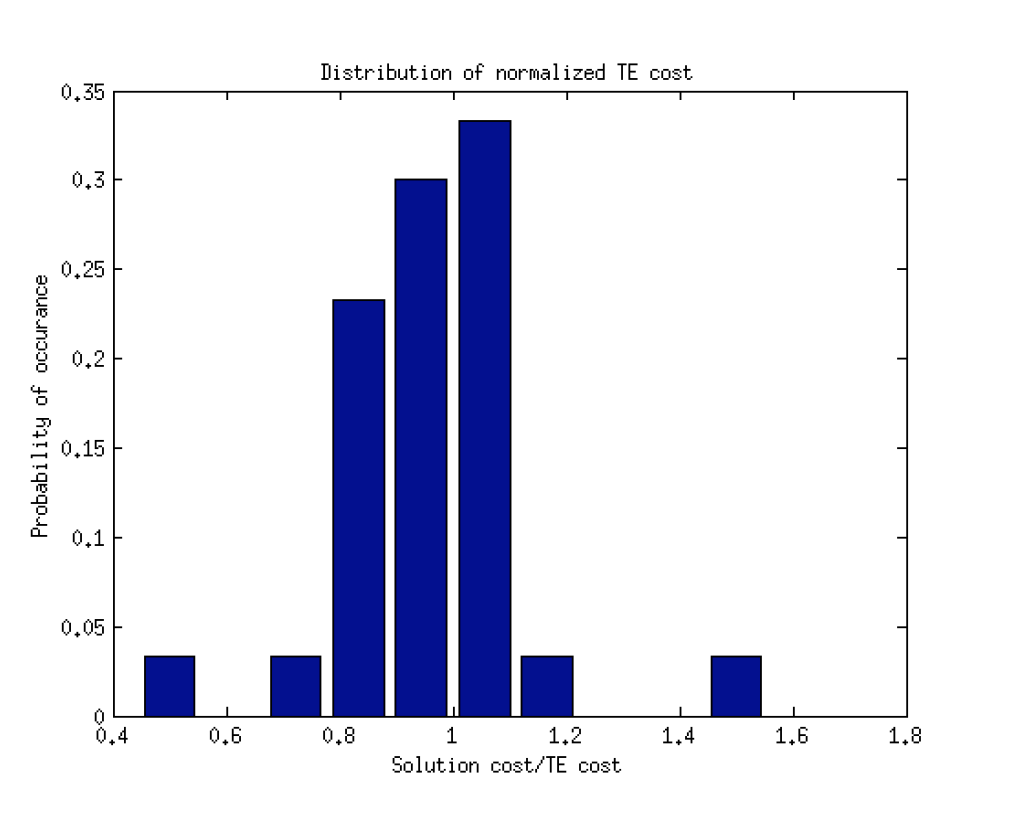}
      \caption{\label{fig:hist} The effect on TE optimality: normalized costs for 29 instances of the repair process, starting from an oscillatory state.}
      \end{figure}
      
   \begin{table}[t]
      \begin{tabular}{ l c r }
      Nodes & Server time (s) & Local search time (s) \\
      \hline
      20 & 0.09020 & 0.350 \\
      30 & 0.2244	&	2.8645\\
      40 & 0.3203 &   3.360\\
      60 & 0.6348 & 17.36 \\
      70 & 0.8418 & 33.59\\
\end{tabular}
     \caption{\label{fig:10} Running time for Waxman random networks. The solver time is reported by Matlab's timer. The local search time is the user and system time before the first search thread terminates.}
\end{table}

Figure \ref{fig:hist} shows that (for the topology of Figure \ref{fig:network}) it did not: in almost all cases the finals cost was close to the initial cost. These figures are normalized to the TE cost of the starting traffic flow, and these states are already optimized using conventional iterative methods.

Oddly, we see that FixRoute is able to improve TE cost in some experiments. We believe that this is because the parallel local search procedure starts from a nearly optimal state and is more likely to find an optimal solution compared to the TE optimization of Fortz and Thorup ~\cite{fifth}.

Our previous results validate the capabilities of FixRoute to repair oscillations with TE optimizations. Our results demonstrate that our heuristic based numeric solution not only can enforce  a particular mandated path preference, but is effective at mitigating the impact on TE. 

To assess the scalibility of our tool, random topologies of various sizes were generated using Waxman's random graph model ~\cite{ninteenth}, using $\alpha=0.15$ and $\beta=0.5$. According to this model, nodes are scattered uniformly at random within a rectangle (here, a unit square). The probability with which a link exists is proportional to its physical length. This probability is given by 

\begin{equation}
P(i,j)=\alpha e^{-d_{ij}/BL},
\end{equation}

where $d_{ij}$ is the distance between two nodes, and $L=\sqrt{2} \max_{i,j} d_{ij}$. Bandwidths on the links are assigned based on their Euclidean length. Links whose length is less than $L/2$ have a bandwidth of $25000$ Mbps, and those of greater length are assigned a bandwidth of $10000$ Mbps.

We also want to assess the effect on running time of giving the mandated preferences. For each test case, we chose $40$ paths consisting of random nodes and of random length, and for those having the same source node (but not necessarily the same egress point), the direction of preference between them was again chosen uniformaly at random \footnote{Note that there are also cases in which neither path is required to be preferred to the other.}.

Though equation \eqref{eq:opt1} can be solved most efficiently using the method of ~\cite{twelvth}, we initially used the CVX 2.0 solver ~\cite{twenty}, ~\cite{twenty1}, incorporating Gurobi to deal with integer constraints ~\cite{twenty2}. This however resulted in long computation time as the size of the network increased above $15$ nodes. Therefore, for larger networks, we used the relaxation of the integer problem. The resulting feasible solution was rounded to an integral one, and this was used as input to the heuristic local search. Note that as the weights are a means of inducing ordering among path, it is highly likely that all of the mandated preferences will remain in effect after this rounding is performed, if not one can return to the original integer programing formulation. Also note that our penalty function $h(x)$ induces a cost to \textit{any} change in $w^e$, and thus the new problem will change the minimum possible number of paths.

       \begin{figure}[t]
  \centering
    \includegraphics[width=0.5\textwidth]{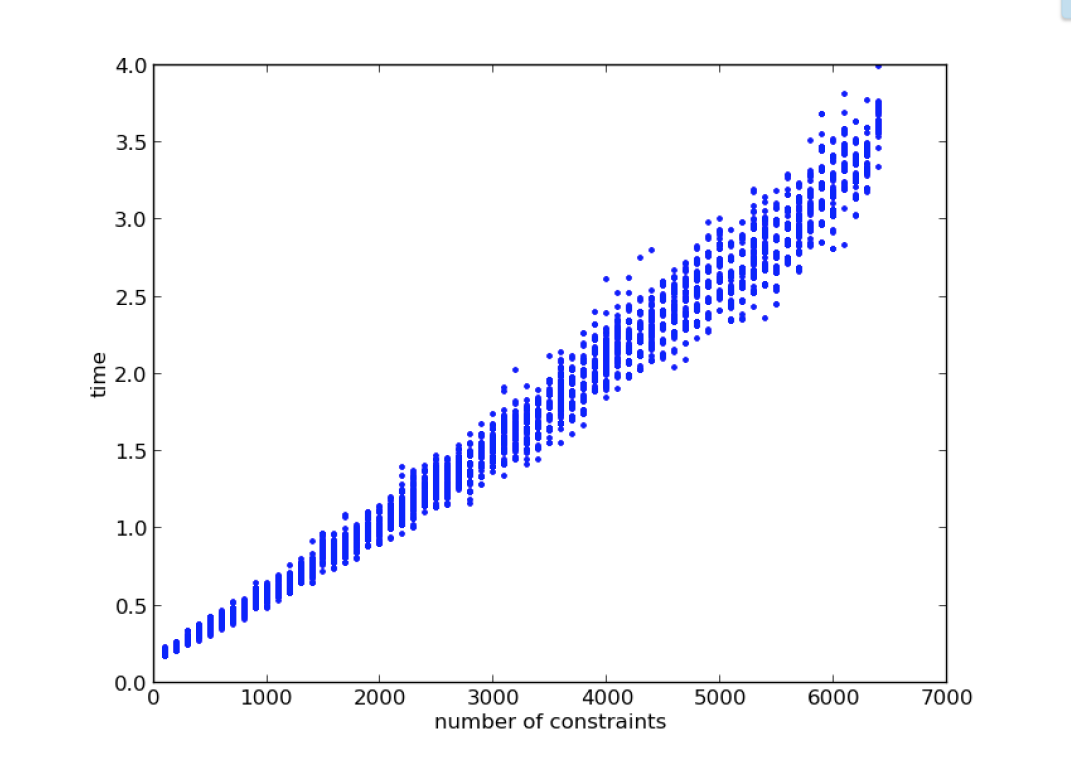}
      \caption{\label{fig:scale} Engine time (s) for different predicate size.}
      \end{figure}
Table \ref{fig:10} shows the running time of the two steps of the local search process, on networks of different sizes. Our running times are obtained by executing FixRoute on a 24-core Linux machine with $66$ GB memory. Each core is an Intel Xeon CPU X5660 at $2.80$ GHz. We observe that as the size of the network increases from 20 to 70, the solver execution time is modest: requires at most $0.82$ \footnote{This is the sum of the execution time required in preparing all four threads.} seconds even for the largest size network. The local search time on the other hand is higher reaching $33.6$ seconds.
As a point of comparison we also perform a similar performance measurement of the Max-SMT approach on the same machine. We note that the Max-SMT solver requires 282 seconds to complete, which is orders of magnitude worse compared to our local search approach. 

 Finally, Figure \ref{fig:scale} demonstrates an apparently linear relationship between the number of constraints and solver time. This shows that Solver time is small regardless of the number of constraints\footnote{Here we add preference constraints uniformly at random. }.
 
 We have done further analysis of the interplay of tradeoffs (between the number of constraints and TE cost) which we have omitted here due to space considerations, a technical report of those results exist which we will not refer to here due to anonymity.

\section{Related Work}
BGP oscillation is well-studied, with theoretical work in the development of SPP and related formalisms ~\cite{sixth}, ~\cite{seventh},~\cite{eigth}. There are related approaches to path computation in distributed systems that have solid basis in linear algebra ~\cite{twenty4}, ~\cite{twenty5}, ~\cite{twenty6}, ~\cite{twenty7} and therefore may be related to `traditional' TE approaches  ~\cite{twenty8}. These have faced their own problems with applicability to IP networks, and various solutions have been developed to handle the nuaunces of techniques such as OSPF ~\cite{fifth}, ~\cite{twenty9}, ~\cite{thirty}, and even BGP ~\cite{thirty1},~\cite{thirty2}.

In network management of routing protocols, recent attention has been given to migration from one IGP to another while avoiding transient faults ~\cite{thirty3} and to avoid MED related problems in iBGP ~\cite{fourth}. Unlike the latter piece of work, we require no protocol change, but rewrite the configuration to meet traffic engineering goals. 

Much of the current activity in routing and network management is centered on the concept of the software-defined network (SDN). In this context, we are greatly influenced by  ~\cite{thirty4} which provides a verified concept of safe network transitions for SDNs. Our work is in the traditional IP context, where dynamic protocols are present, and we are able to supply formally supported methodology for network changes in this domain. Though not the focus of our paper, FixRoute can be applied to SDN configurations as well, and is a basis for our future work. 

\section{conclusion}
This paper presents FixRoute, a toolchain that automatically repairs existing router configurations, so that oscillations are avoided while not penalizing existing traffic engineering goals. Our solution uses a combination of logical and numeric methods, made possible through our novel PSPP model for partial network configurations, integrated with non-convex integer constraint optimizations and local search heuristics. Our performance on Emulab and in large scale network simulations demonstrate that FixRoute is computationally efficient, enables fast network repair at low convergence times, and achieves good traffic engineering performance.

We believe this approach shows promise, and has applicability beyond the immediate use in network repair. In particular rapid execution of the solution process, without the need for multiple explicit simulations, has the potential to transform the management of networks by equipping operators with a higher-level view of configuration possibilities and trade-off.